\documentclass[]{article}
\usepackage{amsmath, amsthm, amssymb}
\usepackage{authblk}
\usepackage{ bm, bbm}
\usepackage{tikz}

\newtheorem{lemma}{Lemma}
\newtheorem{proposition}{Proposition}
\newtheorem{definition}{Definition}
\newtheorem{theorem}{Theorem}

\newtheorem{example}{Example}
\newtheorem{assumption}{Assumption}
\newcommand{\be}{\begin{equation}}
\newcommand{\ee}{\end{equation}}
\newcommand{\ba}{\begin{array}}
\newcommand{\ea}{\end{array}}

\newcommand{\mc}{\mathcal}
\newcommand{\eps}{\varepsilon}
\newcommand{\R}{\mathbb{R}}

\DeclareMathOperator{\argmax}{argmax}
\DeclareMathOperator{\argmin}{argmin}
\newcommand{\0}{\mathbbm{0}}
\newcommand{\1}{\mathbbm{1}}
\newcommand{\A}{\mathcal A}

\renewcommand{\S}{\mc S}

\newcommand{\N}{\mathcal{N}}
\newcommand{\X}{\mathcal{X}}
\newcommand{\V}{\mathcal V}

\newcommand{\ov}{\overline}
\newcommand{\se}{\text{ if }}
\newcommand{\z}{\mathbf{z}}
\newcommand{\y}{\mathbf{y}}
\newcommand{\x}{\mathbf{x}}

\theoremstyle{remark}\newtheorem{remark}{Remark}

\title{Optimal targeting in super-modular games\footnote{This work was partially supported by MIUR grant Dipartimenti di Eccellenza 2018--2022 [CUP: E11G18000350001], the Swedish Research Council, and by the Compagnia di San Paolo.}}
\author{Giacomo Como\thanks{Email: giacomo.como@polito.it. G.~Como is also with the Department of Automatic Control, Lund University, BOX 118, SE-22100, Lund, Sweden. }}
\author{Stephane Durand\thanks{Email: stephane.durand@polito.it}}
\author{Fabio Fagnani\thanks{Email: fabio.fagnani@polito.it}}
\affil{Department of Mathematical Sciences ``G.L.~Lagrange'', Politecnico di Torino, Corso Duca degli Abruzzi 24, 10129, Torino, Italy.}

\begin{document}

%
%
%
%
%
%



\maketitle
\begin{abstract}We study an optimal targeting problem for super-modular games with binary actions and finitely many players. The considered problem consists in the selection of a subset of players of minimum size such that, when the actions of these players are forced to a controlled value while the others are left to repeatedly play a best response action, the system will converge to the greatest Nash equilibrium of the game. Our main contributions consist in  showing that the problem is NP-complete and in proposing an efficient iterative algorithm with provable convergence properties for its solution. We discuss in detail the special case of network coordination games and its relation with the notion of cohesiveness. Finally, we show with simulations the strength of our approach with respect to naive heuristics based on classical network centrality measures.
\end{abstract}

\section{Introduction}

In a game with multiple Nash equilibria, what is the minimum number of players to target in order to force the system to move from an original Nash equilibrium $A$ to a desired Nash equilibrium $B$? This paper deals with such a problem for the class of super-modular games with binary actions and where the two Nash equilibria $A$ and $B$ are, respectively, the least and the greatest in the game.  Our contribution is twofold: we show that the problem is NP-complete and we propose the design of an iterative algorithm for an efficient solution. 

The considered problem can be framed in the more general setting of studying minimal intervention strategies needed to drive a multi-agent system governed by agents' myopic utility maximization to a desired configuration. In applications where the goal is to achieve a social optimum, such interventions are often modeled as perturbations of the utility functions that lead to a modification of the Nash equilibria of the game. This viewpoint is natural for instance in analyzing the effect of taxes or subsidies in economic models or prices and tolls in transportation systems. More recently, a similar approach has been proposed in the context of network quadratic games \cite{Galeotti.Golub.Goyal:2017} to model incentive interventions for instance in school and economic systems.

A different viewpoint, that is the one considered in this paper, is that of individuating a subset of nodes (hopefully small) that if suitably controlled will lead the entire system to the desired equilibrium. The minimum cardinality of this set can also be interpreted as a measure of resilience of the system's equilibrium: the larger it is, the more energy is needed by an external intervention to destabilize it. In the context of binary actions $\{0,1\}$ considered in this paper, the control action simply amounts to force the set of chosen players, originally playing action $0$ state, to play action $1$. This well models situations where action $1$ indicates the use of a certain technology or the adoption of a new product and the control action corresponds for instance to a marketing intervention where, at the targeted individuals, a certain item is offered for free. 

Super-modular games have received a great deal of attention in the recent years as the basic way to model strategic complementarity effects \cite{gamesonnetworks}. Its variegate applications include modeling of social and economic behaviors like adoption of a new technology, participation to an event, provision of a public good effort. They are typically endowed with multiple Nash equilibria that admit a Pareto ordering and the problem of the minimal effort needed to push the system from a lower to a greater equilibrium is natural and relevant in all these applicative contexts.

A fundamental example of super-modular games is that of coordination games over networks. The binary coordination game is analyzed in detail in \cite{morris} where the key concept of cohesiveness of a set of players is introduced and then used in order to characterize all Nash equilibria. Moreover, the question if an initial seed of influenced players (that maintain action $1$ in all circumstances) is capable of propagating to the all network is addressed in the same paper and an equivalent characterization of this spreading phenomenon is also expressed in terms of cohesiveness. 

This contagion phenomenon is the content of our analysis in the more general framework of super-modular games. A subset of nodes from which propagation is successful is called a sufficient control sets and our goal is to find such sets of minimum possible cardinality. We notice that the condition proposed in \cite{morris} is computationally quite demanding and in practice it cannot be used directly to solve the optimization problem even for medium size games. Indeed, even to determine if a single set  is a sufficient control set, it requires a number of check growing exponentially in the cardinality of the complement of such set. 

The complementary problem of understanding (for binary coordination games) what is the maximum possible spreading of the state $1$ starting from an initial seed of a given number $k$ of targeted players, was studied in a seminal paper by \cite{Kempe}. While their problem and ours are related, they are independent, in the sense that solving one does not provide a solution of the other. Another point worth stressing is that, in their setting, \cite{Kempe} consider players equipped with random independent activation thresholds and chose to optimize the expected size of the maximum spreading. They prove that such functional is sub-modular and then they design a greedy algorithm for obtaining sub-optimal solutions. The randomness that they introduce is actually crucial in their approach, as the functional considered is not sub-modular for deterministic choices of thresholds. This lack of sub-modularity is actually a key feature of coordination games where the utility functions present a threshold behavior and make it unfeasible to try to approximate our targeting problem by iteratively adding target nodes in a greedy way. 

A targeting intervention problem, related to the one studied in \cite{Kempe}, is considered in \cite{Galeotti.Goyal:2009}. There, the authors consider the problem of a firm that sells a good to a set of individuals organized through a social network. The firm, in order to maximize its profit, chooses a set of individuals on which to concentrate its advertising efforts or other marketing strategies relative to that specific good. The role of the social network is either of propagating information (in a gossip pairwise style) regarding the good so to push other people to buy it, or rather to model a positive externality effect where the utility of an individual to buy that product depends on the number of neighbors already using it. This second instance is particularly related to the problem studied in \cite{Kempe} with the important difference that here authors model the network in a mean field fashion only considering the degree distribution.

A different targeting intervention problem is studied in \cite{Ballesterand.Calvo.Zenou:2006} where authors consider network quadratic games and individuate the $k$ most influential players by studying how the aggregate output decreases when this set of players is removed from the network. 

The general problem of determining the best set of nodes to exert the most effective control in a networked system has recently appeared in other contexts. In \cite{Liu.Slotine.Barabasi:2011, Gao.ea:2014, Pasqualetti.Zampieri.Bullo:2014} this is studied in the context of controllability problems for general linear network systems. In \cite{Yildiz-2013}, \cite{Vassio-2014}, \cite{Grabisch.Rusinowska:2018} authors focus on the problem of the optimal position of stubborn influencers in voter models or in linear opinion dynamics.  

Our main contribution is twofold. First, we prove that the proposed problem is NP-complete, reducing it to the well known $3$-SAT problem. Second, we design an iterative randomized search algorithm with provable convergence properties towards sufficient control sets of minimum cardinality. The core of the algorithm is a time-reversible Markov chain over the family of all sufficient control sets that starts with the full set, moves through all of them in an ergodic way, and concentrates its mass on those of minimum cardinality.\footnote{A preliminary version of the second part of our results for the special case of network coordination games and not containing any complexity analysis were presented at the 21st IFAC World Congress and published in its proceedings \cite{IFAC-2020}.} 

The rest of the paper is organized as follows. In the final part of this section we report some basic notation used throughout the paper. Section \ref{sec:model} is dedicated to the formal introduction of the problem and in particular of the concept of sufficient control sets. Here we introduce the important notion of monotone improvement path (appeared for other purposes in \cite{Drakopoulos.ea:2015,Drakopoulos.ea:2016}) and we give an equivalent (but more operative) characterization of sufficient control sets. Section \ref{sec:hard} is dedicated to the complexity analysis: we show that the problem is equivalent to an instance of the $3$-SAT problem and thus NP complete. In Section \ref{sec:MC} we present and analyze a distributed algorithm to find optimal sufficient control sets and, in Section \ref{sec:numerical-simulations}, we present some simulation results. Finally, a conclusive Section \ref{sec:conclusions} ends the paper.

We conclude this introduction with a few notational conventions to be adopted throughout the paper.  Vectors are indicated in bold-face letters ${\bf x},\, {\bf y}, {\bf z}$. 
We define the binary vectors $\delta_i$: $(\delta_i)_i=1$ and $(\delta_i)_j=0$ for every $j\neq i$.
For a subset $\mc S\subseteq\{1,\dots , n\}$, we put $\1_{\mc S}=\sum_{i\in\mc S}\delta_i$.
Every ${\bf x}$ in $ \{0,1\}^n$ can be written as ${\bf x}=\1_S$ for some $S\subseteq \{1,\dots , n\}$.  
We use the notation $\1$ for the all-$1$ vector.

\section{Problem formulation and basic properties}\label{sec:model}
We consider finite strategic form games with set of players $\V=\{1,\ldots,n\}$ whereby each player $i$ choses her action $x_i$ from a binary set $\A=\{0,1\}$. Let $\mathcal{X}=\mathcal{A}^n$ denote the (strategy) profile space, whose elements $\bf x$ will be referred to as (strategy) profiles. We shall consider the standard partial order on the strategy profile space $\mc X$, given by 
\be\label{eq:partial-ordering}\x\leq\y\qquad\Longleftrightarrow \qquad x_i\leq y_i,\quad \forall i\in\mc V\,.\ee
As customary, given a strategy profile $\x$  in $\mathcal{X}$ and a player $i$, we indicate with $\x_{-i}$ the strategy profile of all players but $i$. 
Each player $i$ is endowed with a utility function $u_i:\X\to\R$, so that  $$u_i(\x)=u_i(x_i, \x_{-i})$$ denotes the utility of player $i$ when she plays action $x_i$ while the rest of the players' strategy profile is $\x_{-i}$. A game will be formally identified by the triple $(\mc V,\mc A,\{u_i\})$.

The best response for a player $i$ in $\V$ is captured by the set-valued function 
$$\mc B_i(\mathbf{x}_{-i})=\argmax_{a \in \mathcal{A}} u_i(a, \x_{-i})\,,$$
while the set of pure strategy Nash equilibria is formally defined by
$$\N=\{{\bf x}\in\X\,|\, x_i\in \mathcal{B}_i(\x_{-i})\,\forall i\in\V\}\,.$$

Throughout the paper, we shall consider games  satisfying the following \emph{increasing difference property}  \cite{Milgrom.Roberts:1990}. 

\begin{assumption}\label{assumption:increasing-differences}
For every player $i$ in $\V$ and every two strategy profiles $\x,\y$ in $\mc X$ such that $\x_{-i}\geq \y_{-i}$, 
\be\label{super-modular} 
u_i(1,\x_{-i}) -u_i(0,\x_{-i})\geq u_i(1,\y_{-i}) -u_i(0,\y_{-i})\,.
\ee
\end{assumption}

Assumption \ref{assumption:increasing-differences} states that the marginal utility of increasing player $i$'s action from $x_i=0$ to $x_i=1$ is a non-decreasing function of the strategy profile $\x_{-i}$ of all the other players. 
For finite games, as is our case, such increasing difference property is equivalent to \emph{super-modularity} \cite{Topkins:1979,Vives:1990,Topkins:1998}. For this reason, we will refer to a game $(\mc V,\mc A,\{u_i\})$ satisfying (\ref{assumption:increasing-differences}) as to a finite super-modular game. In the economic literature, these are also referred to as games of \emph{strategic complements} \cite{Milgrom.Shannon:1994}. 

A standard result for super-modular games ensures that their set of pure strategy Nash equilibria is always nonempty and there exist a minimal and a maximal Nash equilibria with respect to the partial order \eqref{eq:partial-ordering}. Throughout the paper, we shall assume  that such minimal and maximal pure strategy Nash equilibria are the all-$0$ profile $\0$ and, respectively, the all-$1$  profile  $\1$. This assumption implies no effective loss of generality since the presence of players that maintain a strict preference for action $0$ or action $1$
independently from the actions played by the other players can be easily integrated  in our framework by suitably modifying the other players' utilities.

In this paper, we study the problem of finding subsets of players $\mc S\subseteq\mc V$  of minimal cardinality for which there exists an improvement path from $\mc S$ to the whole player set $\mc V$. This is formalized by the following definitions.

\begin{definition} For a finite game with binary actions $(\mc V,\mc A,\{u_i\})$, a sequence of strategy profiles $({\bf x}^k)_{k=0,\dots , m}$ is  an \emph{improvement path} from the set $\mc S\subseteq \mc V$ to the set $\mc T\subseteq \mc V$ if 
\begin{enumerate}
\item ${\mathbf x}^0=\1_{\S}$, ${\mathbf x}^m=\1_{\mc T}$
\item for every $k=0,\dots ,m-1$ there exists $i_k$ in $\mathcal{V}\setminus\S$ such that 
\begin{itemize}
\item ${\mathbf x}^{k+1}_{-i_k} ={\mathbf x}_{-i_k}^k$ and ${x}^{k+1}_{i_k} \neq{x}^k_{i_k}$
\item $u_{i_k}({\mathbf x}^{k+1})\geq u_{i_k}({\mathbf x}^{k})$
\end{itemize}
\end{enumerate}
\end{definition}

\begin{definition}[Sufficient control set]\label{def:sufficient} For a finite game with binary actions $(\mc V,\mc A,\{u_i\})$, 
\begin{itemize}
\item $\mc S\subseteq \mc V$ is a \emph{sufficient control set} if there exists an improvement path from $\mc S$ to $\mc V$.
\item A sufficient control set $\mc S\subseteq \mc V$ is \emph{optimal} if there exists no sufficient control set of strictly smaller cardinality.
\end{itemize}
\end{definition}
Notice that sufficient control sets always exist, as the whole set of players $\mc V$ trivially is a sufficient control set. Our objective is to find optimal sufficient control sets.

A key fact is that, in dealing with the concept of sufficient control set, it is not restrictive to consider exclusively improvement paths where all action changes are from $0$ to $1$. Such improvement paths are formally defined below.

\begin{definition}[Monotone Improvement path]\label{def:crusade} 
For a finite game with binary actions $(\mc V,\mc A,\{u_i\})$, an improvement path  $({\bf x}^k)_{k=0,\dots , m}$ from the set $\mc S\subseteq \mc V$ to the set $\mc T\subseteq \mc V$ is called \emph{monotone } if there exists a sequence of distinct players $i_k$ in $\mc T\setminus\mc S$ for $k=0,\dots , m-1$ such that ${\mathbf x}^{k+1} ={\mathbf x}^k+\delta_{i_k}$ for $k=0,\dots , m-1$. 
\end{definition}

\begin{remark}\label{remark:monotone} 
Notice that a monotone improvement path from $\mc S$ to $\mc T$ is completely specified by the sequence of players $i_k$ in $\mc T\setminus\mc S$, $k=1,\ldots,m$, which are sequentially changing their actions from $0$ to $1$. Observe that 
$\mathcal{T}\setminus\S=\{i_1,\dots ,i_m\}$ and thus  the path length $m=|\mathcal{T}\setminus\S|$ coincides with the difference between the cardinality of the arrival set $\mc T$ and the one of the departure set $\mc S$. 
\end{remark}

The following result formalizes our previous claim. 

\begin{lemma}\label{lemma:valid-sufficient} In a finite super-modular game with binary actions $(\mc V, \mc A, \{u_i\})$, $\mc S\subseteq \mc V$ is a {sufficient control set} if and only if there exists a monotone improvement path from $\mc S$ to $\mc V$.
\end{lemma}
\begin{proof}
Clearly, if there exists a monotone improvement path from $\mc S$ to $\mc V$, then $\mc S$ is a sufficient control set. 

Conversely, if $\S$ is a sufficient control set, then there exists a (not necessarily monotone) improvement path $({\bf y}^k)_{k=0,\dots , T}$
in $\mathcal{X}$ from $\mc \S$ to $\V$.
For every player $i$ in $\V\setminus\S$, define
$$k(i)=\min\{k=1,\dots , T\;|\; {\bf y}^{k}={\bf y}^{k-1}+\delta_{i}\}$$
that is the first time that player $i$ changes her action from $0$ to $1$ along the path $({\bf y}^k)_{k=0,\dots , T}$. Now, let $m=n-|\S|$ and order the players in $\V \setminus\S$ as $i_1,\dots , i_m$ in such a way that $k({i_1})<k({i_2})<\cdots <k({i_m})$. 
Then, for every $h=0,1,\ldots,m$, define
$${\bf x}^h=\1_{\S}+\sum_{j=1}^{h}\delta_{i_j}$$
and notice that ${\bf x}^{h-1}\geq {\bf y}^{k(i_h)-1}$. Using the increasing difference property we now obtain that 
$$u_{i_h}({\bf x}^{h})- u_{i_h}({\bf x}^{h-1})\geq 
u_{i_h}(\y^{k(i_h)})- u_{i_h}(\y^{k(i_h)-1})\geq 0\,,$$
for every $h=1,\ldots,m$. 
This shows that $({\bf x}^k)_{k=0,\dots , m}$
is an improvement path from $\S$ to $\mc V$.  
By construction, this improvement path is also monotone, thus proving the claim. 
\end{proof}

This new characterization of sufficient control sets, allows for proving the following intuitive fact.
\begin{proposition}[monotonicity for inclusion]\label{prop:superset}
\label{mono}
In a finite super-modular game with binary actions $(\mc V, \mc A, \{u_i\})$, if $\mc S\subseteq \mc V$ is a {sufficient control set} then every  $\mc T\subseteq\mc V$ such that $\mc S\subseteq\mc T$ is also a {sufficient} control set.
\end{proposition}
\begin{proof}
Assume that $\S$ is a {sufficient} control set and let $\mc T\supseteq \S$.
Because of Lemma \ref{lemma:valid-sufficient}, there exists a monotone improvement path $({\bf x}^k)_{k=0,\dots , m}$ 
from $\S$ to $\mc V$. Consider the associated sequence of players $(i_k)$ for $k=1,\dots , m$ such that
${\bf x}^{k+1}- {\bf x}^k=\delta_{i_k}$ for each $k$.
Consider now the subsequence of points $i_{k_1}, i_{k_2}, \dots , i_{k_{m'}}$ that are in $\mathcal{V}\setminus\mc T$ and put $\y^h=\max\{\1_{\mc T}, \x^{k_h}\}$ for $h=1,\dots , m'$. By construction, we have that $\y^{h}\geq \x^{k_{h+1}-1}$. By the increasing difference property (\ref{super-modular}) and the fact that $({\bf x}^k)_{k=0,\dots , m}$ is a monotone improvement path from $\S$ it follows that, for every $h$, putting $i=i_{k_{h+1}}$,
$$\begin{array}{rcl}u_i(\y^{h+1})-u_i(\y^{h})&=&u_i(1, \y^{h}_{-i})-u_i(0, \y^{h}_{-i})\\[5pt]
&\geq& u_i(1, \x^{k_{h+1}-1}_{-i})-u_i(0, \x^{k_{h+1}-1}_{-i})\\[5pt]
&=&u_i(\x^{k_{h+1}})-u_i(\x^{k_{h+1}-1})\\[5pt]
&\geq& 0\end{array}
$$
This says that $({\bf y}^k)_{k=0,\dots , m'}$ is a monotone improvement path from $\mc T$ to $\mc V$.
\end{proof}

\begin{remark} The notion of sufficient control set introduced in Definition \ref{def:sufficient} can be reinterpreted in terms of the asynchronous best response dynamics. Given a subset $\mc S\subseteq \mc V$, consider the Markov chain $X^t$ on the strategy profile space $\mc X$ 
whose transitions are described as follows. At every discrete time, a player, among those in $\mc V\setminus \mc S$, is chosen uniformly at random and  updates her played action choosing uniformly at random among the actions of her current best response to the other players' strategy profile. Notice that the existence of an improvement path from $\mc S$ to $\mc V$ is equivalent to say that, for every initial state $X^0$ such that $X^0_i=1$ for all $i$ in $\mc S$, the Markov chain $X^t$ will reach the all-$1$ profile $\1$ in finite time with positive probability. 

Actually,  more is true. Consider any strategy profile $\bf x$ in $\mc X$ such that $X^0_i=1$ for all $i$ in $\mc S$, equivalently such that $\bf x=\1_{\mc S'}$ for some superset $\mc S'\supseteq\mc S$.  If $\mc S$ is a sufficient control set, it follows from Proposition \ref{prop:superset} that also $\mc S'$ is a sufficient. This implies that there exists a monotone improvement path from $\mc S'$ to $\mc V$ and thus $X^t$ will also reach $\1$ from $\x$ in finite time with positive probability. If the all-1 strategy profile $\1$ is a strict Nash equilibrium (in the sense that all players have, in that profile, a best response consisting of the singleton $1$) then this argument proves that $\mc S\subseteq \mc V$ is a sufficient control set if and only if the corresponding Markov chain $X^t$ is absorbed in $\1$ in finite time with probability one. In the more general case, if there are players for which $0$ and $1$ are always indifferent independently from the behavior of the other players, then the condition on the Markov chain is replaced by 
the existence of a set of profiles containing $\1$ on which the Markov chain $X^t$ gets trapped in finite time with probability one and within such set it moves ergodically. 
\end{remark}

\section{Optimal targeting in network coordination games}\label{sec:coordination}
A notable example of super-modular games with binary actions is that of network coordination games. In this section, after reviewing  this class of games, we study the optimal targeting problem for two special instances. We first study coordination games on arbitrary undirected networks where the players have homogeneous thresholds characterizing their best responses and we highlight, for this case, the connection of our problem with the notion of cohesiveness \cite{morris}. The second case we consider is that of coordination games on a complete graph with heterogeneous thresholds for which, we show that the optimal targeting problem admits a relatively simple analytical solution (c.f.~\cite{Granovetter,Rossi.ea:2019}).

Let $\mc G=(\mc V,\mc E,W)$ be a finite weighted directed graph, whereby $\mc V$ is the set of nodes, $\mc E\subseteq\mc V\times\mc V$ is the set of directed links, and $W$ in $\R_+^{\mc V\times\mc V}$ is the weight matrix, such that $W_{ij}>0$ if and only if there is a link $(i,j)$ in $\mc E$ directed from its tail node $i$ to its head node $j$. A positive entry $W_{ij}$ of the weight matrix $W$ represents the weight of the link $(i,j)$. Let $w_i=\sum_{j\ne i}W_{ij}$ denote the out-degree of a node $i$ in $\mc V$. We shall assume that $\mc G$ contains no self-loops, equivalently, that the diagonal elements of the weight matrix $W$ are all zero, and no sinks, i.e., that $w_i>0$ for every $i$ in $\mc V$. We shall refer to the graph $\mc G$ as simple if $W_{ij}=W_{ji}$ in $\{0,1\}$ (in this case $W$ is completely determined by $\mc E$).

A \emph{network coordination game} on a graph $\mc G=(\mc V,\mc E,W)$ is a game $(\mc V, \mc A, \{u_i\})$ 
with binary action set $\mc A=\{0,1\}$ and utilities  
\be\label{utilities-coordination}u_i(x)=\sum_{j\ne i}W_{ij}\left((1-x_i)(1-x_j)+x_ix_j\right)+c_i x_i\,,\ee
where the constant $c_i$ in $[-w_i,w_i]$ models a possible bias of player $i$ towards action $0$ (if $c_i<0$) or action $1$ (if $c_i>0$). In fact, the best response functions are given by 
\be\label{bestrespons-coord}\mc B_i(x_{-i})=\left\{\ba{lcl}\{0\}&\se&\frac1{w_i}\sum_{j\ne i}W_{ij}x_j<\theta_i\\[7pt]
\{0,1\}&\se&\frac1{w_i}\sum_{j\ne i}W_{ij}x_j=\theta_i\\[7pt]
\{1\}&\se&\frac1{w_i}\sum_{j\ne i}W_{ij}x_j>\theta_i\ea \right. \ee
where
\be\label{threshold-general}\theta_i=\frac{w_i-c_i}{2w_i}\ee
is the threshold of player $i$ in $\mc V$.  In the special case when the graph is simple and $c_i=0$ (so that the threshold is $\theta_i=1/2$) for every player $i$ in $\mc V$, this is also known as the \emph{majority game}.

\subsection{Homogeneous network coordination games}

In this subsection, we focus on the special case when the players all have the same threshold $\theta_i=\theta$ in $[0,1]$. Sufficient control sets in this case can be equivalently formulated in terms of the graph-theoretic notion of cohesiveness introduced in \cite{morris}. Specifically, a subset of nodes $\mc S\subseteq \mc V$ is called $\alpha$-\emph{cohesive} in a graph $\mc G$ if 
\be\label{cohesiveness}\sum_{j\in\mc S}W_{ij}\ge\alpha w_i\,,\qquad \forall i\in\mc S\,,\ee
For a simple graph, the above means that every node in $\mc S$ has at least a fraction $\alpha$ of its neighbors within $\mc S$ (equivalently, at most a fraction $1-\alpha$ outside $\mc S$).  Considerations in \cite{morris} and in \cite{gamesonnetworks} yield the following characterization of sufficient control sets. Define a subset $\mc S\subseteq \mc V$ \emph{uniformly no more than $\theta$-cohesive} if no subset $\mc S'\subseteq \mc S$ is $\theta'$-cohesive for some $\theta'>\theta$. The following is a consequence of this definition and explicitly proven in \cite{gamesonnetworks} (see Proposition 4 therein).
\begin{proposition}\label{prop:morris} $\mc S\subseteq \mc V$ is a sufficient control set for the network coordination game where all players have threshold $\theta$ if and only if $\mc V\setminus\mc S$ is uniformly no more than $(1-\theta)$-cohesive.
\end{proposition}
This reformulation of the concept of sufficient control set is of limited interest from the computational point of view. Indeed, checking that the set $\mc V\setminus\mc S$ is uniformly no more than $(1-\theta)$-cohesive involves an analysis of all possible subsets of $\mc V\setminus\mc S$. Nevertheless, this characterization can be used to analyze specific cases. 

Below we present examples of sufficient control sets for the special case of the majority game for specific simple connected graphs.

\begin{example} In this example the game we are considering is always the majority game (namely the network coordination game where all players have threshold $1/2$) on a simple and connected graph $\mc G=(\mc V,\mc E, W)$.
\begin{itemize}

\item  Let $\mc G$ be the complete graph with $n$ nodes. Then, all subsets consisting of $\lfloor n/2\rfloor$ nodes are optimal sufficient control sets. This is because every subset of $\lceil n/2\rceil$ or fewer nodes is not $\theta$-cohesive for any $\theta>1/2$. This says that every subset of $\lceil n/2\rceil$ nodes is uniformly no more than $1/2$-cohesive and result follows from Proposition \ref{prop:morris}.
Optimality follows directly from the fact that smaller subsets are never sufficient control sets.
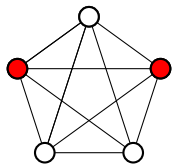
\begin{figure}
\centering
\begin{tikzpicture}
\foreach \i in {90,162,...,450}
{
	\draw (\i+144:1)--(\i:1) -- (\i+72:1);
}
\foreach \i in {90,162,...,450}
{
	\draw[white,fill] (\i:1) circle (.13);
	\draw[thick] (\i:1) circle (.13);
	
}
	\draw[red,fill] (162:1) circle (.13);
	\draw[thick] (162:1) circle (.13);

	\draw[red,fill] (18:1) circle (.13);
	\draw[thick] (18:1) circle (.13);

\end{tikzpicture}
\caption{An optimal sufficient control set for the complete graph}
\label{fig-clique}
\end{figure}

\item Let $\mc G$ be a simple connected graph where every node has degree at most $2$. Then, every set consisting of a single node is a sufficient control set (and is automatically optimal).  To see this, notice that every strict subset of players $\mc S \subsetneq\mc V$ must possess a node $i$ in $\mc S$ with $W_{ij}>0$ for some $j$ in $\mc V\setminus\mc S$ (otherwise the graph would not be connected). This implies that the set $\mc S$ cannot be $\theta$-cohesive for any $\theta>1/2$. We can then conclude as in the previous item. 
An instance is depicted in Figure \ref{fig-degreetwo}.

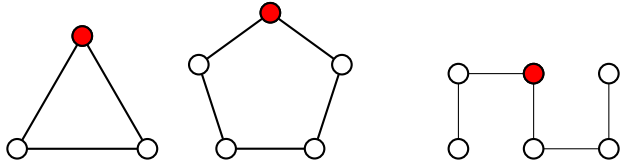
\begin{figure}
\centering
\begin{tikzpicture}
\begin{scope}{yshift=.86 cm}
	\draw[thick] (90:1)--(330:1)--(210:1)--(90:1);

\foreach \i in {330,90,210}
{
	\draw[white,fill] (\i:1) circle (.13);
	\draw[thick] (\i:1) circle (.13);
}
\draw[red,fill] (90:1) circle (.13);
\draw[thick] (90:1) circle (.13);
\end{scope}
\begin{scope}[xshift=2.5 cm, yshift=.31cm]

\foreach \i in {90,162,234,306,18}
{
	\draw[thick] (\i:1)--(\i+72:1);
}
\foreach \i in {90,162,234,306,18}
{
	\draw[white,fill] (\i:1) circle (.13);
	\draw[thick] (\i:1) circle (.13);
}
\draw[red,fill] (90:1) circle (.13);
\draw[thick] (90:1) circle (.13);

\end{scope}
\begin{scope}[xshift=5 cm, yshift=-.5 cm]
\draw (0,0)--(0,1)--(1,1)--(1,0)--(2,0)--(2,1);
\foreach \i in {0,1,2}
{
\foreach \j in {0,1}
{
	\draw[white,fill] (\i,\j) circle (.13);
	\draw[thick] (\i,\j) circle (.13);
}
}
\draw[red,fill] (1,1) circle (.13);
\draw[thick] (1,1) circle (.13);

\end{scope}

\end{tikzpicture}
\caption{Optimal sufficient control sets for graphs with nodes of degree at most $2$.}
\label{fig-degreetwo}
\end{figure}

\item Let $\mc G$ be a tree. Then, the set of the leaf nodes is always a sufficient control set. Indeed let $\mc S$ be any subset of the nodes not containing leaves and consider a path (a walk with no repeated nodes) of maximum length all consisting of nodes in $\mc S$, say $(i_1,\dots , i_l)$. Notice that $i_1$ can not have other neighbors in $\mc S$ otherwise the path could be extendable. On the other hand, since $i_1$ is not a leaf in the tree, it must have degree at least $2$, namely, at least one neighbor outside of $\mc S$. This implies that $\mc S$ can not be $\theta$-cohesive for $\theta>1/2$. We conclude using again Proposition \ref{prop:morris}.
In general, such sets are not optimal. Indeed, the argument above shows that also the set of nodes that are neighbors of the leaves is a sufficient control set, typically of smaller cardinality than the set of leaves. 
An example is reported in Figure  \ref{fig-tree}.

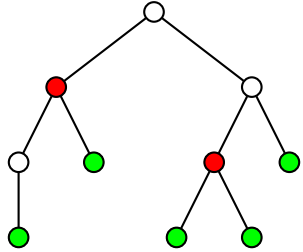
\begin{figure}
\centering
\begin{tikzpicture}

\draw[thick] (-1.3,3) -- (0,4);
\draw[thick] (1.3,3) -- (0,4);
\draw[thick] (-1.8,2) -- (-1.3,3);
\draw[thick] (-1.8,1) -- (-1.8,2);
\draw[thick] (-.8,2) --(-1.3,3);
\draw[thick] (.8,2) --(1.3,3);
\draw[thick] (1.8,2) --(1.3,3);
\draw[thick] (0.3,1) --(.8,2);
\draw[thick] (1.3,1) --(.8,2);

\draw[fill,white] (0,4) circle (.13);
\draw[thick] (0,4) circle (.13);

\draw[fill,red] (-1.3,3) circle (.13);
\draw[thick] (-1.3,3) circle (.13);

\draw[fill, white] (1.3,3) circle (.13);
\draw[thick] (1.3,3) circle (.13);

\draw[fill,white] (-1.8,2) circle (.13);
\draw[thick] (-1.8,2) circle (.13);

\draw[fill,green] (-1.8,1) circle (.13);
\draw[thick] (-1.8,1) circle (.13);

\draw[fill,green] (-.8,2) circle (.13);
\draw[thick] (-.8,2) circle (.13);

\draw[fill,red] (.8,2) circle (.13);
\draw[thick] (.8,2) circle (.13);

\draw[fill,green] (1.8,2) circle (.13);
\draw[thick] (1.8,2) circle (.13);

\draw[fill,green] (0.3,1) circle (.13);
\draw[thick] (0.3,1) circle (.13);

\draw[fill,green] (1.3,1) circle (.13);
\draw[thick] (1.3,1) circle (.13);

\end{tikzpicture}
\caption{Two examples of sufficient control sets for a tree: the one consisting of the leaves in green and the one consisting of the neighbors of the leaves in red. This second one is optimal.}
\label{fig-tree}
\end{figure}

\item Let $\mc G$ be the $d$-dimensional grid graph having node set $\mc V$ and link set $\mc E$ respectively given by 
$$\mc V=\{0,\dots , k-1\}^d\,,\qquad \mc E=\left\{({\bf a},{\bf b})\in\mc V\times\mc V\,|\, \sum\nolimits_{h=1}^k|a_h-b_h|=1\right\}\,.$$
Put $\mc S_l=\{(a_1,\dots , a_d)\in\mc V\;|\; \sum a_i=l\}$. We claim that $\mc S_{k-1}$ is a sufficient control set. To see this, notice that any ${\bf a}$ in $ \mc S_l$ has exactly $d$ neighbors in $\mc S_{l+1}$ if $l<k-1$. Similarly, any ${\bf a}$ in $ \mc S_l$ has exactly $d$ neighbors in $\mc S_{l-1}$ if $l>k-1$. Considering that the degree of every node is at most $2d$ in $\mc G$, a simple induction argument then allows to construct a monotone improvement path from $\mc S_{k-1}$ to the whole of $\mc V$.
It can be checked directly that this control set is optimal for $d=1$ and $d=2$, while is not for $d\geq 3$.

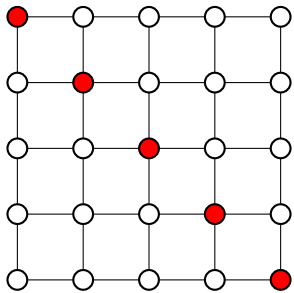
\begin{figure}
\centering
\begin{tikzpicture}[scale=.875]
\foreach \i in {0,...,4}
{
	\draw (\i,0)--(\i,4);
	\draw (0,\i)--(4,\i);
}
\foreach \i in {0,...,4}
{

	\foreach \j in {0,...,4}
	{
		\draw[white,fill] (\i,\j) circle (.15);
	}
	\draw[red,fill] (\i,4-\i) circle (.15);

	\foreach \j in {0,...,4}
	{
		\draw[ thick] (\i,\j) circle (.15);
	}
}

\end{tikzpicture}
\caption{An optimal sufficient control set for a $2$-dimensional grid.}
\label{fig-grid}
\end{figure}

\end{itemize}
\end{example}

The examples considered above show that optimal sufficient control sets for the majority game may exhibit different relative sizes depending on the considered graph. In complete graphs, their size is a constant fraction of the number $n$ of players and we expect the same to hold in very well connected graphs as for instance random Erdos-Renji graphs. This conjecture is corroborated by numerical simulations presented in Section \ref{sec:numerical-simulations}. In contrast, for more loosely connected graphs (trees, grids), the size of optimal sufficient control sets scales as a negligible fraction of the size $n$.

\subsection{Heterogeneous coordination game on the complete graph}
In this subsection, we focus on network coordination games on the complete graph, whereby $W_{ij}=1$ for every $i\ne j$ in $\mc V$. In contrast with the previous subsection, we shall allow for full heterogeneity of the players' thresholds, that in this case are given by 
$$\theta_i=\frac{n-c_i-1}{2(n-1)}\,,\qquad i\in\mc V\,.$$
Our results show that optimal sufficient control sets can be completely characterized in terms of the threshold distribution function 
$$F(z)=\frac1n\left|\left\{i\in\mc V:\,\theta_i\le z\right\}\right|\,,\qquad z\in[0,1]\,.$$

First, we have the following technical result. 

\begin{lemma}\label{lemma-complete}
Consider a heterogeneous network coordination game on the complete graph with threshold distribution $F(z)$. 
Then, $\emptyset$ is a sufficient control set if and only if 
\be\label{distr-cond}F(z)\ge z\,,\qquad \forall z\in[0,1]\,.\ee
\end{lemma}
\begin{proof}
We start with a general consideration that will be used to prove both implications. Fix $\mc S\subseteq\mc V$ and let $\x=\1_{\mc S}$. Put $n_1=|\mc S|$. 
It follows from (\ref{bestrespons-coord}) that for every player $i$ such that $x_i=0$, it holds
$\mc B_i(\x_{-i})=\{0\}$ if only if 
\be\label{thetai>}\theta_i>\frac{n_1}{n-1}\,.\ee

Suppose $\emptyset$ is not a sufficient control set and let $\mc S\subsetneq\mc V$ be a set of maximum cardinality for which there exists a monotone improvement set from $\emptyset$ to $\mc S$. Put $x=\1_{\mc S}$ and $n_1=|\mc S|\leq n-1$. It follows from previous consideration that all players $i$ such that $x_i=0$, have a threshold $\theta_i$ satisfying (\ref{thetai>}). Therefore,
$$n-n_1\le\left|\left\{i\in\mc V:\,\theta_i>\frac{n_1}{n-1}\right\}\right|=n\left(1-F\left(\frac{n_1}{n-1}\right)\right)$$
By dividing both sides by $n$ and rearranging terms, we obtain 
\be\label{F<=}F\left(\frac{n_1}{n-1}\right)\le\frac{n_1}n<\frac{n_1}{n-1}\,.\ee
This implies that (\ref{distr-cond}) does not hold true.

Suppose instead that (\ref{distr-cond}) does not hold true and let $z$ in $ [0,1]$ be such that $F(z)<z$. By the way $F$ is defined, there exists $n_1$ in $\{0,1,\ldots,n-1\}$ such that $F(z)=n_1/n$. Observe that $n_1=nF(z)<zn$ implies that $n_1\le zn-1$ and, consequently,
$$\frac{n_1}{n-1}\le\frac{zn-1}{n-1}\le z\,.$$ 
Then, by monotonicity of  the threshold distribution function we get 
\be\label{F<=bis}F\left(\frac{n_1}{n-1}\right)\le F(z)=\frac{n_1}{n}\,.\ee
Let $\mc S$ be a set consisting of $n_1$ players with the least possible threshold and let $\x=\1_{\mc S}$. 
It  then follows from \eqref{F<=bis} that each player $i$ playing $x_i=0$ has threshold satisfying \eqref{thetai>} and hence, as observed at the beginning of this proof, it is such that $\mc B_i(x_{-i})=\{0\}$. 
This implies that there cannot be a monotone improvement path from $\mc S$ to $\mc V$.  Consequently $\mc S$  is not a sufficient control set and neither is $\emptyset$ by Proposition \ref{prop:superset}. 
\end{proof}

As an application of Lemma \ref{lemma-complete} we obtain the following characterization of the optimal sufficient control sets for heterogeneous network coordination games on the complete graph.

\begin{proposition}
Consider a heterogeneous network coordination game on the complete graph with threshold distribution $F(z)$. 
Then, the minimal size of a sufficient control set is 
$$M=\left\lceil n\cdot \sup_{0\le z\le 1}\left[z-F(z)\right]_+\right\rceil\,.$$
In particular, every $\mc S$ consisting of $M$ players $i$ in $\mc V$ with the $M$ largest thresholds $\theta_i$ gives an optimal sufficient control set. 
\end{proposition}
\begin{proof}
First observe that a subset of players $\mc S\subseteq\mc V$ is a sufficient control set for the network coordination game with utilities \eqref{utilities-coordination} if and only if $\emptyset$ is a sufficient control set for the modified network coordination game with utilities 
\be\label{modified-utilities}\ov u_i(x)=\left\{\ba{lcl}
\sum_{j\ne i}((1-x_i)(1-x_j)+x_ix_j)+(n-1)x_i &\se& i\in\mc S\\
\sum_{j\ne i}((1-x_i)(1-x_j)+x_ix_j)+c_ix_i &\se& i\in\mc V\setminus\mc S\,,\ea
\right.\ee
whereby all the players $i$ in $\mc S$ have modified threshold $\ov\theta_i=0$ and the rest of the players $j$ in $\mc V\setminus\mc S$ have the same threshold $\ov\theta_j=\theta_j$. Let $\ov F(z)$ be the threshold distribution function of this modified game and observe that 
\be\label{newthreshold}0\le\ov F(z)-F(z)\le |\mc S|/n\,,\qquad \forall z\in[0,1]\,.\ee 

We now show that any subset $\mc S\subseteq \mc V$ such that $|\mc S|<M$ can not be a sufficient control set. If $M=0$ there is nothing to prove. Assume now that $M\geq 1$ and notice that we can write 
\be\label{ceiling}n\cdot\sup_{0\le z'\le1}\left[z'-F(z')\right]_+=M-1+\epsilon\ee
for some $\epsilon >0$. Notice that, since $\sup_{0\le z'\le1}\left[z'-F(z')\right]_+>0$, we have that 
\be\label{sup}\sup_{0\le z'\le1}\left[z'-F(z')\right]_+=\sup_{0\le z'\le1}\left[z'-F(z')\right]\ee
If $|\mc S|<M$, then (\ref{newthreshold}), (\ref{ceiling}), and (\ref{sup}) imply that, for every $z$ in $ [0,1]$,
$$0\le n\left(\ov F(z)-F(z)\right)\le|\mc S|\le  M-1=n\cdot\sup_{0\le z'\le1}\left[z'-F(z')\right]-\eps\,,$$
This yields 
$$  z-\ov F(z)\ge z-F(z)-\sup_{0\le z'\le1}\left[z'-F(z')\right]+\eps/n\,,\qquad \forall z\in[0,1]\,.$$
and taking the $\sup$ on both sides, we finally obtain
$$\sup_{0\le z\le1}\left[z-\ov F(z)\right]\geq \eps/n>0$$
Then, Lemma \ref{lemma-complete} implies that $\emptyset$  is not a sufficient control set for the modified network coordination game with utilities \eqref{modified-utilities} and, consequently, $\mc S$ is not a sufficient control set for the original game. 

To complete the proof, we now consider a set $\mc S$ of $M$ players with the highest thresholds. In this case,
$$\ov F(z)=\min\{1,F(z)+M/n\}\ge\min\left\{1,F(z)+\sup_{0\le z'\le 1}\left[z'-F(z')\right]_+\right\} \ge z\,,$$ 
for every $z$ in $[0,1]$. It then follows from Lemma \ref{lemma-complete} that $\emptyset$  is a sufficient control set for the modified network coordination game with utilities \eqref{modified-utilities}, thus showing that $\mc S$ is a sufficient control set for the original game. 
\end{proof}

\section{Complexity of finding a sufficient control set}\label{sec:hard}

In this section, we study the complexity of finding sufficient control sets for arbitrary super-modular games and prove that it is an NP-complete problem \cite[Section 7.4]{Sipser:2012}.  

Formally, given a binary super-modular game and a positive integer $n$ we define $SCS$ to be the logical proposition "there exists a sufficient control set of size less then or equal to $s$ for the game".

\begin{theorem}\label{theorem:NP-completeness}
The problem $SCS$ is NP-complete. 
\end{theorem}
In order to prove Theorem \ref{theorem:NP-completeness}, we will first show that $SCS$ belongs to the complexity class NP (c.f., \cite[Definition 7.19]{Sipser:2012}) and then that it is NP-hard.

\begin{lemma}\label{lemma:NP}
The problem $SCS$ belongs to NP.
\end{lemma}
\begin{proof}
We show that, given an instance of a finite binary-action super-modular game and a witness consisting in subset of players $\mc S\subseteq\mc V$, checking if $\mc S$ is a sufficient control set can be done in a time growing proportionally to the square of $n-s$, where $n=|\mc V|$ and $s=|\mc S|$. In fact, this can be achieved by an iterative algorithm that starts with time index $t=0$ and profile $\x(0)=\1_{\mc S}$ and then proceeds as follows. If there exists at least one player $i$ in $\mc V$ such that \be\label{eq:improve}x_i(t)=0\,,\qquad 1\in\mc B_i(x_{-i}(t))\,,\ee
then arbitrarily chose one such player $i$, increase the time index $t$ by one unit and  define the new profile $\x(t)$ with $x_{i}(t)=1$ and $x_{-i}(t)=x_{-i}(t-1)$. Otherwise, if no player $i$ satisfying  \eqref{eq:improve} exists, then halt and return the current value of the time index $t$. Since, by Proposition \ref{prop:superset}, every superset of a sufficient control set is itself a sufficient control set, we have that $\mc S$ is a sufficient control set if and only if the algorithm defined above terminates with $t=n-s$. Clearly, the number of steps of the algorithm is at most $n-s$ and at the $t$-th step, it is necessary to compute the best responses of at most $n-t$ players, so that the algorithm effectively requires at most $\sum_{t=0}^{n-s-1}(n-t)=(n-s)(n-s+1)/2$ best response computations. 
This proves that the problem belongs to the complexity class NP.  \end{proof}

We will now prove that $SCS$ is NP-hard by showing that the 3-SAT problem \cite[Ch.~7.2]{Sipser:2012} can be reduced, in polynomial time, to a particular instance of $SCS$. 
Consider any instance $I=(X,C)$ of  the 3-SAT problem, consisting of a set of variables $X = \{x_1 ,x_2 \cdots x_{s-1}\}$ and clauses $C = \{c_1,c_2,\cdots c_m\}$, such that in every clause in $C$ exactly three, possibly negated, variables from $X$ appear. Then, we associate to $I$ a simple graph  $\mc G_I=(\mc V_I,\mc E_I)$ of order $|\mc V_I|=2s+5m$ and size $|\mc E_I|=s+8m$ as follows.  
The node set $\mc V_I$ is the union of the following six disjoint sets of nodes:
\begin{itemize}
\item A set $\mc W=\{w_1,w_2,\ldots,w_m\}$, whose elements correspond each to a clause in $C$; 
\item A set $\mc Y=\{y_1,y_2,\ldots, y_{s-1}\}$, whose elements correspond each to a variable in $I$, with the interpretation that $y_i$ encodes $x_i$ if $x_i$ is true;
\item A set $\bar{\mc Y}=\{\bar y_1,\bar y_2,\ldots,\bar y_{s-1}\}$, whose elements correspond each to a variable in $I$, with the interpretation that $\bar y_i$ encodes $x_i$ if $x_i$ is false; 
\item A single node $z$, whose role will be to break possible ties; 
\item Two sets of leaves $\mc L$ and $\mc M$, of cardinality $|\mc L|=3m$ and $|\mc M|=m+1$. 
\end{itemize}
Links in $\mc E_I$ only connect pairs of nodes belonging to different sets and in particular:
\begin{enumerate}
\item[$(1)$] A node $w_j$ in $\mc W$ is connected to a node $y_i$ in $\mc Y$ if and only if the variable $x_i$ appears in the clause $c_j$ and to a node $\bar{y}_i$ in $\bar{\mc Y}$ if and only if the variable $\bar{x}_i$ appears in the clause $c_j$; 
\item[$(2)$] For each clause containing the variable $x_i$, node $y_i$ in $\mc Y$ is connected to a different node in $\mc L$, and for  each clause containing the variable $\ov x_i$, node $\ov y_i$ in $\ov{\mc Y}$ is connected to a different node in $\mc L$, in such a way that the elements of $\mc L$ are each connected to exactly one element  either of $\mc Y$ or of $\ov{\mc Y}$; 
\item[$(3)$] The node $z$ is connected to every element of $\mc W$ and of  $\mc M$; 
\item[$(4)$] For every $i=1,\ldots,s-1$, node $y_i$ is connected to the corresponding node $\bar{y_i}$. 
\end{enumerate}
There is a total of $3m$ links of type $(1)$,  $3m$ links of type (2), $2m+1$ links of type (3),  and $s-1$ links of type $(4)$. Nodes in 
$\mc L$ and $\mc M$ all have degree $1$, nodes in $\mc W$ all have degree $4$, node $z$ has degree $2m+1$, while the degree of a node $y_i$ in $\mc Y$ (respectively $\ov y_i$ in $\ov{\mc Y}$) is $1$ plus twice the number of clauses the variable $x_i$ (respectively, $\ov x_i$) appears in. 

\begin{figure}\centering
\begin{tikzpicture}
\draw (140:.4)--(140:1.6);
\draw (140:1) node[above,right]{$\,(1)$};
\draw (220:.4)--(220:1.6);
\draw (220:1) node[below,right]{$\,(1)$};
\draw (0:.4)--(0:1.6);
\draw (1,0) node[above]{$(3)$};
\draw (3,0) node[above]{$(3)$};
\draw (2.4,0)--(3.6,0);

\draw (140:2)++(-.4,0)--++(-1.2,0);
\draw (140:2)++(0,-.4)--(220:2);
\draw (-1.4,0) node[left]{$(4)$};
\draw (220:2)++(-.4,0)--++(-1.2,0);
\draw (140:2)++(-1,0) node[above]{$(2)$};
\draw (220:2)++(-1,0) node[above]{$(2)$};

\draw (140:2) circle (.4);
\draw (140:2) node{$\mc Y$};

\draw[white,fill] (220:2) circle (.4);
\draw (220:2) circle (.4);
\draw (220:2) node{$\bar{\mc Y}$};
\draw (0,0) circle (.4);
\draw (0,0) node{$\mathcal W$};
\draw (0:2) circle (.4);
\draw (0:2) node{$z$};
\draw (-3.4,0) ellipse (.5 and 1.5);
\draw (-3.4,0) node{$\mathcal{L}$};
\draw (0:2)++(2,0) circle (.4);
\draw (0:2)++(2,0) node{$\mathcal{M}$};
\end{tikzpicture}

\end{figure}

Now, we shall consider the majority game on the graph $\mc G_{I}$, whereby each player in $\mc V_I$ has action set $\{0,1\}$ and the utility of player $i$ is equal to the number of her neighbors that play the same action as her. We then ask the question "is there a sufficient control set of size less then or equal to $s$ for this game?"
We will now show that the answer to this question is true if and only if the instance of 3-SAT is satisfiable. 

\begin{lemma}\label{lemma:SATimpliesCONTROL}
Let $I=(X,C)$ be an instance of the 3-SAT problem, and let $\mc G_I=(\mc V_I,\mc E_I)$ be the simple graph defined above. 
If $I$ is satisfiable with a solution $x^*$ in $\{0,1\}^{s-1}$, then 
$$\mc S=\{z\}\cup\{y_i:\,x^*_i=1\}\cup\{\ov y_i:\,x^*_i= 0\}$$ 
is a sufficient control set of size $s$ for the majority game on $\mc G_{I}$. 
\end{lemma}
\begin{proof} 
Since $I$ is satisfied by $x^*$, for every clause $c_j$ in $C$ there exists $i$ in $\{1,\ldots,s-1\}$ such that either $x_i$ appears in $c_j$ and $x_i^*=1$ or $\ov x_i$ appears in $c_j$ and $\ov x_i^*=1$. Thus, in the graph $\mc G_I$, all clause-related nodes in $\mc W$ have at least one neighbor in $(\mc Y\cup\ov{\mc Y})\cap\mc S$. Since they are all connected to $z$ in $\mc S$ also, and have all degree $4$ in $\mc G_I$, this implies that there exists a monotone improvement path from $\mc S$ to $\mc S\cup \mc W$.

Now, consider a variable $x_i$ in $X$ and let $m_i$ be the number of clauses it appears in. 
Then, notice that, if the corresponding node $y_i$ in $\mc Y$ does not belong to $\mc S$, it necessarily has one neighbor in $\mc S$ ($\ov y_i$) as well as $m_i$ neighbors in $\mc W$ (those corresponding to the clauses it belongs to).
Since its degree in $\mc G_I$ is exactly $2m_i+1$, this implies that  $\mc S\cup \mc W\cup\mc Y$ can be reached by a monotone improvement path  from $\mc S\cup \mc W$, hence from $\mc S$. Analogously, one proves that $\mc S\cup \mc W\cup\mc Y\cup\ov{\mc Y}$ can be reached by a monotone improvement path from $\mc S$. 

Finally, since every remaining node in $\mc L\cup\mc M$ is of degree one and connected to a node in $\mc Y\cup\ov{\mc Y}\cup\{z\}$, we get that the monotone improvement path from $\mc S$ can be extended to reach the whole node set $\mc V_I$, thus proving that $\mc S$ is a sufficient control set. 
\end{proof}

We will now show that the converse of Lemma \ref{lemma:SATimpliesCONTROL} holds true. 
\begin{lemma}\label{lemma:CONTROLimpliesSAT}
 Let $I=(X,C)$ be an instance of the 3-SAT problem, and let $\mc G_I=(\mc V_I,\mc E_I)$ be the simple graph defined above. 
 If there is a sufficient control set $\mc S$ of size $s$ for the majority game on $\mc G_{I}$, then $I$ is solvable.
  \end{lemma}
  \begin{proof} 
  We will first show that there exists a sufficient control set $\mc S'$ of the same size $s$ containing $z$ and exactly one node between $y_i$ and $\bar{y}_i$ for $1\le i\le s-1$. We argue as follows. First, notice that, for every $i=1,\dots , s-1$, at least one node among $y_i$, $\bar y_i$, and the leaves in $\mc L$ connected to them must be in $\mc S$ for, otherwise, it is easy to check that no improvement path would ever be able to reach the pair $\{y_i,\bar y_i\}$. Similarly, at least one element among $z$ and the leaves in $\mc M$ must be in $\mc S$.\\
In case when neither $y_i$ nor $\bar y_i$ belong to $\mc S$, removing the leaf connected to them that is in $\mc S$ and adding its sole neighbor (either $y_i$ or  $\bar y_i$) maintains the control set sufficient and preserves its size. We construct $\mc S'$ in this way replacing leaves with variable nodes and finally applying the same substitution idea to include the node $z$ removing a leaf connected to it.

Now observe that, because of the structure of the graph and since $\mc S'$ contains no leaves in $\mc L\cup\mc M$, in any monotone improvement path from $\mc S'$ to $\mc V_I$, a node in $(\mc Y\cup\ov{\mc Y})\setminus{\mc S'}$ can only appear after all nodes in $\mc W$ have already appeared. Since all nodes in $\mc W$ have degree $4$, this says that each of them must have at least two neighbors in $\mc S'$.
This implies that every node in $\mc W$ must have at least one neighbor in $\mc S'\setminus\{z\}\subseteq \mc Y\cup\ov{\mc Y}$. 

Consider now the candidate solution $x^*$ in $\{0,1\}^{s-1}$ that has $x^*_i=1$ if and only if $y_i $ in $\mc  S'$. Then, it follows from the argument above that for every clause $c_j$ there exists $i$ in $\{1,\ldots,s-1\}$ such that either $x_i$ appears in $c_j$ and $x_i^*=1$ or $\ov x_i$ appears in $c_j$ and $\ov x_i^*=1$. This proves that $I$ is solvable. 
  \end{proof}

Lemma \ref{lemma:SATimpliesCONTROL} and Lemma \ref{lemma:CONTROLimpliesSAT} thus show that starting from an instance of the 3-SAT, we could build an instance of the SCS problem in polynomial time and of polynomial size, whose answer is the same as that of the 3-SAT. This shows that SCS is NP-hard. Together with Lemma \ref{lemma:NP}, this implies that SCS is an NP-complete problem.

\section{A distributed algorithm for optimal control sets}\label{sec:MC}
The characterization of sufficient control sets through the concept of monotone improvement paths (Lemma \ref{lemma:valid-sufficient}) suggests the possibility that such sets may be searched for by starting from the all-$1$ profile $\1$ and iteratively replacing $1$'s with $0$'s in the attempt to follow backwards a monotone improvement path. In  order to capture this intuition, in this section we introduce a family of discrete-time Markov chains $(Z_t^{\eps})_{t\ge0}$ on the strategy profile space $\mathcal{X}$, parameterized by a scalar $\eps$ in $[0,1]$. We will then prove that, for $0<\eps\le1$, the Markov chain $(Z_t^{\eps})_{t\ge0}$ is time-reversible and that, as $\eps$ vanishes, its stationary distribution concentrates on the family of optimal sufficient control sets. 

The dynamics of the Markov chain $Z_t^{\eps}$ are described as follows: at every discrete time $t=0,1\,\ldots$, given that $Z_t^{\eps}=\z$, a player $i$ is chosen uniformly at random from the whole player set $\mc V$. Then, if $u_i(1,\z_{-i})< u_i(0,\z_{-i})$, the state is not changed, i.e., $Z_{t+1}^{\eps}=\z$.
Otherwise, if $u_i(1,\z_{-i})\ge u_i(0,\z_{-i})$, then if the current action of player $i$ is $z_i=1$ it is changed to $0$ with probability $1$, while if her current action is $z_i=0$, it is changed to $1$ with probability $\eps$.
%
%
The transition probabilities of this Markov chain are then given by
\be\label{Pepsilon}
P^{\epsilon}_{\x,\y}=\left\{\ba{lcl}
1/n&\se& \y=\x-\delta_i\text{ and }u_i(\y)\le u_i(\x)\\
\eps/n&\se& \y=\x+\delta_i\text{ and }u_i(\y)\ge u_i(\x)\\
0&\se&\text{ otherwise }\,,
\ea\right.\ee
for every $\x,\y$ in $\mc X$.

Notice that, for $\eps=0$, only transitions from $1$ to $0$ are allowed. 
In fact, in this case, the Markov chain ${\bf Z}^0_t$ has absorbing states. Specifically, let 
\be\label{eq:Z-def}
\mc Z  = \{x\in\mathcal{X}\;|\; \mathbb{P}(\exists t_0\ge0\,:\, Z^0_{t_0}=x\;|\; {\bf Z}^0_0=\mathbbm{1})>0\}
\ee
be the set of all states that are reachable by the Markov chain $Z_t^0$ when started from ${\bf Z}^0_0=\1$
and let 
\be\label{eq:Z-infty}\mc Z_{\infty}  =  \{x\in\mathcal{X}\;|\; \mathbb{P}(\exists t_0\ge0\,:\, Z^0_{t}=x \,\forall t\geq t_0\;|\; Z^0_0=\mathbbm{1})>0\}\ee 
be the set of absorbing states reachable by $Z_t^0$ from ${\bf Z}^0_0=\1$. 
We have the following result. 

\begin{proposition}\label{prop:Z} For a finite super-modular game with binary actions $(\mc V, \mc A, \{u_i\})$, let $\mc Z$ and $\mc Z_{\infty}$ be defined as in \eqref{eq:Z-def} and \eqref{eq:Z-infty}, respectively. Then,  
\begin{enumerate}
\item[(i)] $\mc S\subseteq\mc V$ is a sufficient control set if and only if $\1_{\mc S}$ in $\mc Z$;
\item[(ii)]  if $\mc S$ is a minimal sufficient control set then $\1_{\mc S}$ in $\mc Z_{\infty}$. 
\end{enumerate}
\end{proposition}
\begin{proof}(i) By definition, $\x=\1_{\mc S}$ belongs to  the set of reachable states $\mc Z $ if and only if there exists a sequence of strategy profiles $(\y^k)_{k=0,\dots , l}$, such that 
${\y}^0=\1$,  ${\y}^l = \1_{\mc S}$, and 
\be\label{eq:yk}{\y}^{k}= {\y}^{k-1}-\delta_{i_k}\,, \qquad u_{i_k}(\y^{k})\le u_{i_k}(\y^{k-1})\qquad1\le k\le l\,.\ee
Notice that \eqref{eq:yk} is equivalent to say that the reversed path $(\x^k)_{k=0,\dots , l}$ with $\x^k=\y^{l-k}$ for $0\le k\le l$ is a monotone improvement path from 
$\mc S$ to $\mc V$. By Lemma \ref{lemma:valid-sufficient}, this is equivalent to say  that $\mc S$ is a sufficient control set.

(ii) If $\mc S$ is a minimal sufficient control set, we know from point (i) that the strategy profile $\1_{\mc S}$ belongs to the reachable set $\mc Z$. Now if, by contradiction, $\1_{\mc S}$ did not belong to the set of reachable absorbing states $\mc Z_{\infty}$, then, from $\x=\1_{\mc S}$, the Markov chain ${\bf Z}^0_t$ could reach, in one step, a different state $\x'=\1_{\mc S'}$ with $\mc S'\subsetneq \mc S$, thus contradicting the minimality assumption on $\mc S$. 
\end{proof}

Point (i) of Proposition \ref{prop:Z} implies that the problem of finding optimal sufficient control sets can be equivalently stated as the problem of finding strategy profiles $\x$ in $\mc Z$ of minimal $l_1$-norm $||\x||_1=\sum_kx_k$, i.e., that $\mc S$ is an optimal sufficient control set if and only if 
$$\1_S\in\argmin\limits_{\x\in\mc Z}||\x||_1\,.$$
Point (ii) implies that we can actually restrict the minimization above to the set $\mc Z_{\infty}$ of absorbing states of the 
Markov chain ${\bf Z}^0_t$ that are reachable from the all-1 strategy profile. However, as the example below shows, $\mc Z_{\infty}$ may contain profiles corresponding to sufficient control sets that are suboptimal and, possibly, not even minimal. 
\begin{example} Consider the majority game on the ring graph with four nodes $\{1,2,3,4\}$. Then, $\z^1=(1,0,1,0)$ in $\mc Z_\infty$ corresponds to the sufficient control set $\mc S=\{1,3\}$, but it is not minimal since $\{1\}$ is also a sufficient control set.
%
%
\end{example}
As a consequence, 
by simply simulating the Markov chain ${\bf Z}^0_t$ started from ${\bf Z}^0_0=\1$, we are not guaranteed to reach an optimal sufficient control set. To overcome this issue, we will instead use the Markov chain $Z^{\eps}_t$ with $\eps >0$, which, as shown below, is time-reversible and ergodic on whole set $\mc Z$ of reachable strategy profiles and, hence, it does not get trapped in non-optimal control sets, and at the same time has a stationary distribution concentrating on the set of optimal control sets as the parameter $\eps$ vanishes. 
\begin{theorem}
For a finite super-modular game with binary actions $(\mc V, \mc A, \{u_i\})$, let $\mc Z$ be defined as in \eqref{eq:Z-def}. 
Then, for $\eps >0$, the Markov chain $Z^{\eps}_t$ with transition probabilities \eqref{Pepsilon}
\begin{enumerate}
\item[(i)] keeps the set $\mc Z$ invariant, namely, if $Z^{\eps}_0$ belongs to $\mc Z$, then $Z^{\eps}_t$ belongs to $\mc Z$ for every $t\ge0$;
\item[(ii)] is time-reversible and ergodic on the set $\mc Z$;
\item[(iii)] has stationary probability 
\be\label{eq:stationary}\mu^\eps _{\x} := \frac1{K_{\eps}}\eps^{ ||\x||_1}\,,\qquad \x\in\mc Z\,,\ee
where $K_{\eps}=\sum_{\x\in\mc Z}\eps^{ ||\x||_1}$. In particular,  $\mu^\eps$ converges to a probability measure $\mu$ concentrated on the set of profiles corresponding to optimal sufficient control sets as $\eps$ vanishes.
\end{enumerate}
\end{theorem}
\begin{proof} 
(i) Let $\x$ in $\mc Z$ be strategy profile that is reachable from the all-$1$ profile by the Markov chain $Z_t^0$ and let $\y$ in $\mc X$ be a strategy profile such that $P^{\eps}_{{\x},\y}>0$. We need to prove that $\y$ belongs to $\mc Z$. If $\y=\x-\delta_i$ for some player $i$ in $\mc V$, then it follows from \eqref{Pepsilon} that $0<P^{\eps}_{{\x},\y}=1/n$ and then $P^{0}_{{\x},\y}=1/n>0$, thus implying that the strategy profile $\y$ belongs to $\mc Z$. 

On the other hand, if $\y=\x+\delta_i$ for some player $i$ in $\mc V$,  we argue as follows. Since $\x$ in $\mc Z$ is a strategy  profile reachable by the Markov chain $Z^0_t$ from the all-1 profile, we can find a sequence of profiles $(\x^k)_{k=0,\dots , l}$ such that $\x^0=\1$ and $\x^l=\x$ and $P^{0}_{{\x^{k-1}},\x^{k}}>0$ for every $k=1, \dots , l$.  From \eqref{Pepsilon}, this is equivalent to  $\x^{k}=\x^{k-1}-\delta_{i_k}$ and $u_{i_k}(\x^k)\leq u_{i_k}(\x^{k-1})$ for some $i_k$ in $\mc V$,  for every $k=1,\dots , l$. Let $s$ in $\{1,\dots , l\}$ be such that $i_s=i$ and consider the sequence $(\z^k)_{k=0,\dots , l-1}$ such that $\z^k=\x^k$ for $k\leq s-1$ and $\z^k=\x^{k+1}+\delta_i$ for $k\geq s$. Notice that, for $k\geq s$,
\be\label{increasing}
\z^k=\x^{k+1}+\delta_i=\x^k+\delta_i -\delta_{i_{k+1}}=\z^{k-1}-\delta_{i_{k+1}}\ee
Relation (\ref{increasing}) and the super-modularity property (\ref{super-modular}) yield 
$$u_{i_{k+1}}(\x^{k+1})\leq u_{i_{k+1}}(\x^{k})\Rightarrow u_{i_{k+1}}(\z^k)\leq u_{i_{k+1}}(\z^{k-1})$$
for every $k\geq s-1$. This implies that $P^{0}_{{\z^{k-1}},\z^{k}}>0$ for every $k=1, \dots , l-1$.  Since
$\z^{l-1}=\x^l+\delta_i=\y$, this proves that the strategy profile $\y$ belongs to $\mc Z$.

(ii) Notice that
\be\label{Peps-1}\eps^{ ||\x||_1}P^{\eps}_{{\x},\y}=\eps^{ ||\y||_1}P^{\eps}_{{\y},\x}\,,\ee
 for every two strategy profiles  $\x$ and $\y$ in $\mc X$. This implies that the Markov chain $Z^{\eps}_t$ is time-reversible with respect to the stationary distribution \eqref{eq:stationary}. 
 
 Since the transitions that have positive probability for the Markov chain $Z^{0}_t$ have also positive probability for the Markov chain $Z^{\eps}_t$, we have that all profiles in $\mc Z$ can be reached from the all-$1$ profile by the Markov chain $Z^{\eps}_t$. Moreover, Equation \eqref{Peps-1} implies that a transition probability $P^{\eps}_{\x,\y}$ is positive if and only if the reverse transition $P^{\eps}_{\y,\x}$ is positive. This implies that $\1$ is reachable from any other profile in $\mc Z$ and thus we conclude that $Z^{\eps}_t$ is ergodic on $\mc Z$. 
 
 (iii) Ergodicity and Equation \eqref{Peps-1} imply that, for every $\eps>0$, the unique stationary distribution of the Markov chain $Z^{\eps}_t$ on the set $\mc Z$ has the form \eqref{eq:stationary}. As $\eps$ vanishes, a direct check shows that the stationary distribution $\mu^{\eps}$ converges to a uniform distribution on the set $\argmin_{{\x}\in\mc Z}||\x||_1$. Using Proposition \ref{prop:Z}, the set $\argmin_{{\x}\in\mc Z}||\x||_1$ coincides with the set of optimal sufficient control sets, thus completing the proof.
\end{proof}

\section{Numerical simulations}\label{sec:numerical-simulations}

In this section, we briefly present some numerical simulations of the proposed algorithm for the case of the majority game on Erd\"os-Renyi random graphs. The Erd\"os-Renyi graph $E(n, p)$ is a random undirected graph with $n$ nodes where  undirected links between pairs of nodes are present with probability $p$ in $[0,1]$ independently  from one another. We consider the order of the graph $n$ ranging up to $70$ and two different scalings for the probability $p$. In the first case, we consider a constant $p=0.4$ independent from the graph order $n$, thus leading to quite a densely connected graph. 
In contrast, in the second case, we choose $p=4\frac{\log n}{n}$, a choice leading to a more sparse graph that nevertheless remains connected with high probability as the graph order $n$ grows large \cite[Theorem 2.8.1]{Durrett:2006}.
We run the randomized algorithm $Z^{\eps}_t$, with $\epsilon =0.3$, for a number of steps proportional to the square of the size of the graph (exactly $100 n^2$) and the control set returned is the one of minimum cardinality during the walk. For small values of $n$, an explicit comparison with the optimal solution, obtained through exhaustive search, proves the efficiency of our approach. Simulations are reported in Figure \ref{fig:erdos}. In Figure \ref{fig:compare} we have made a comparison with respect to a naive heuristics selecting the highest degree nodes. Specifically, for each value of $n$, we have considered the highest degree nodes set of the same cardinality as the one found by our algorithm and we have plotted the percentage of the graph nodes that would turn to $1$ using that specific control set. When $n$ is sufficiently large this percentage is around $30\%$ and shows how the degree is not the right property to look at in the optimization of these control sets.

\begin{figure}
\includegraphics[scale=.5]{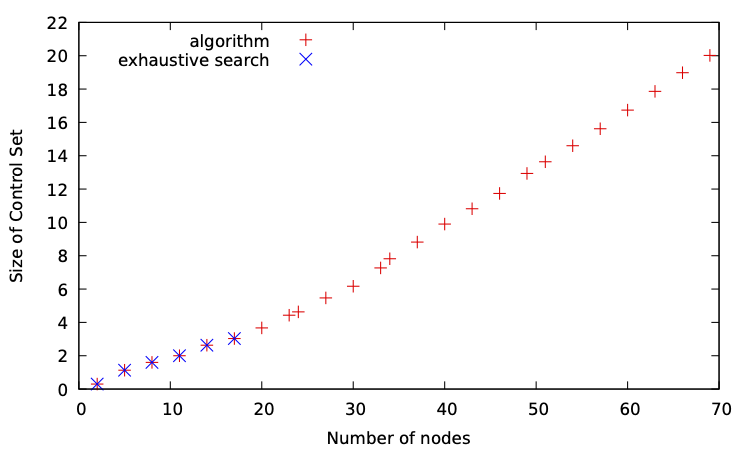} \includegraphics[scale=.5]{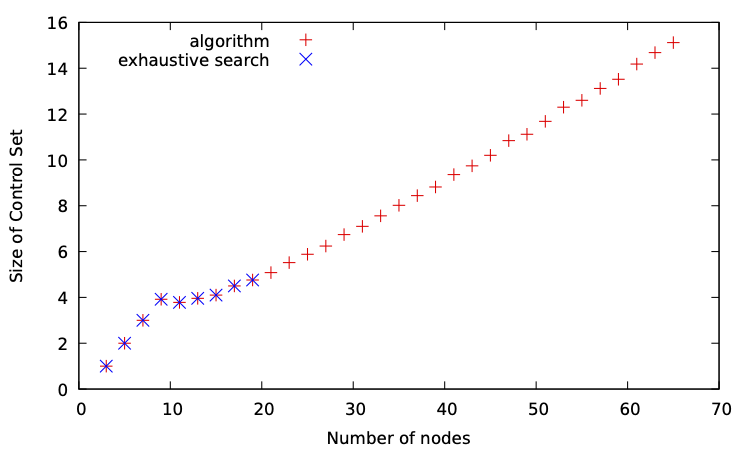}
\caption{Size of Control Sets for random graphs $E(n, p)$ with $p=0.4$ (left) and $p=4\frac{\log n}{n}$ (right) }
\label{fig:erdos}
\end{figure}
\begin{figure}
\includegraphics[scale=.5]{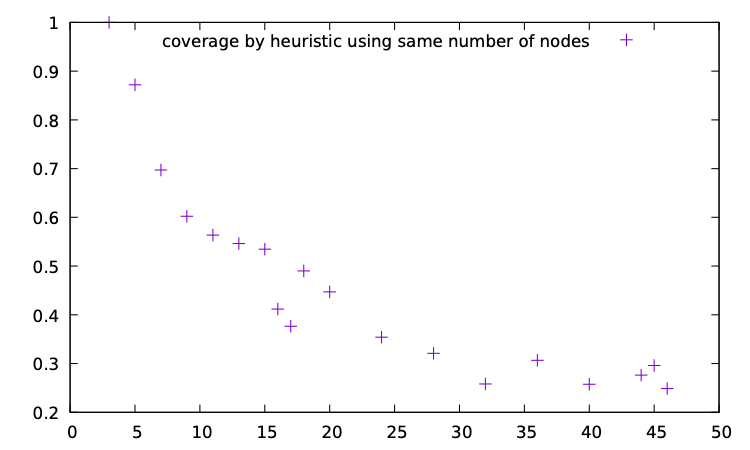} \includegraphics[scale=.5]{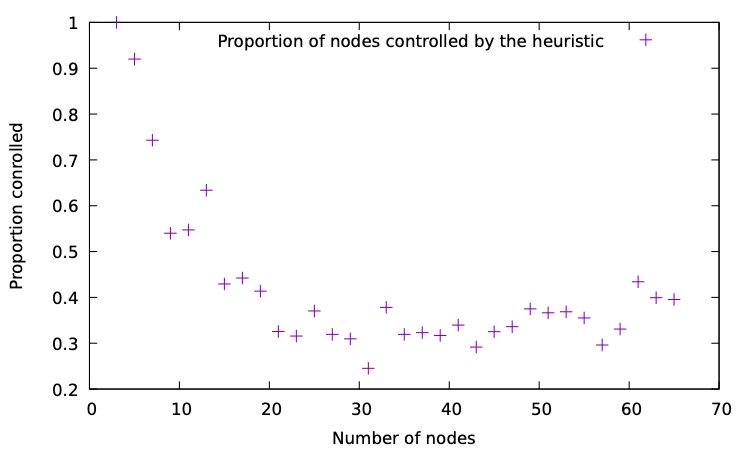}
\caption{Coverage obtained by taking the $k$ highest degree node, with $k$ the size of the set found by the algorithm for random graphs $E(n, p)$ with $p=0.4$ (left) and $p=4\frac{\log n}{n}$ (right)}
\label{fig:compare}
\end{figure}

\section{Conclusion}\label{sec:conclusions}
In this paper, we have studied a novel optimal targeting problem for super-modular games with binary action set and finitely many players. 
The considered problem consists in the selection of a subset of players of minimum size such that, when the actions of these players are forced to the value $1$, there exists a monotone improvement path from the minimal to the maximal pure strategy Nash equilibrium of the constrained super-modular game. 
Our main contributions consist in:  (i) showing that this is an NP-complete problem; (ii) proposing a computationally simple randomized algorithm that provably selects an optimal solution with high probability. Finally, we have presented some numerical simulations for the case of the majority game on Erd\"os-Renyi random graphs. We have compared the performance of our algorithm with that of an exhaustive search (for small problem sizes) and that of a simple heuristic where target players are those with the highest centrality in the graph. The first such comparison validates our theoretical results. The second comparison shows that the centrality-based heuristic performs as much as 70\% worse than our algorithm in this problem, thus highlighting the relevance of our analysis.

The problem studied in this paper can be considered a particular instance of a control problem in a game-theoretic framework. Our results show how the structure of the game, i.e., super-modularity, can be leveraged to get insight into the solution of the control problem. Several directions for future research can be considered. For instance, in the context of super-modular games, natural generalizations include the extension to non-binary action sets and the consideration of possibly more complex actions altering the utilities of the controlled players rather than directly forcing their action to a desired one. Our techniques strongly leverage on the super-modularity assumption. Extensions to more general classes of games are challenging and would likely require the development of different technical tools.


\bibliographystyle{plain} 
\bibliography{bib}

\end{document}